\documentclass[10pt,conference]{IEEEtran} 
\usepackage{times}

\usepackage{graphicx}
\usepackage[linesnumbered,lined,boxed,commentsnumbered,algoruled]{algorithm2e}

\usepackage{amsthm}
\usepackage{mathtools}
\usepackage{listings}
\usepackage{color}
\usepackage{url}
\usepackage{threeparttable}
\usepackage{subfig}
\usepackage{enumitem}
\usepackage{textcomp}
\usepackage{amsfonts}
\usepackage{hyperref}
\usepackage{cleveref}
\definecolor{dkgreen}{rgb}{0,0.6,0}
\definecolor{gray}{rgb}{0.5,0.5,0.5}
\definecolor{mauve}{rgb}{0.58,0,0.82}

\newcommand{\gm[1]}{G(#1)}

\usepackage{balance}  

\begin{document}

\title{Extensible Data Skipping}
\author{\IEEEauthorblockN{Paula Ta-Shma, Guy Khazma, Gal Lushi, Oshrit Feder}
\IEEEauthorblockA{IBM Research \\
Email: \{paula,oshritf\}@il.ibm.com, \{Guy.Khazma,Gal.Lushi\}@ibm.com}}

\newtheorem{thm}{Theorem}
\newtheorem{definition}[thm]{Definition}
\newtheorem{claim}[thm]{Claim}
\newtheorem{problem}[thm]{Problem}
\newtheorem{remark}[thm]{Remark}
\newtheorem{lemma}[thm]{Lemma}
\newtheorem{assumption}[thm]{Assumption}
\newtheorem{case}[thm]{Case}

\maketitle

\begin{abstract}
Data skipping reduces I/O for SQL queries by skipping over irrelevant data objects (files) based on their metadata. We extend this notion by allowing developers to define their own data skipping metadata types and indexes using a flexible API. 
Our framework is the first to natively support data skipping for arbitrary data types (e.g. geospatial, logs) and queries with User Defined Functions (UDFs). 
We integrated our framework with Apache Spark and it is now deployed across multiple products/services at IBM. We present our extensible data skipping APIs, discuss index design, and implement various metadata indexes, requiring only around 30 lines of additional code per index. 
In particular we implement data skipping for a third party library with geospatial UDFs and demonstrate speedups of two orders of magnitude. Our centralized metadata approach provides a x3.6 speed up even when compared to queries which are rewritten to exploit Parquet min/max metadata. We demonstrate that extensible data skipping is applicable to broad class of applications, where user defined indexes achieve significant speedups and cost savings with very low development cost.  
\end{abstract}

\section{Introduction}
According to today's   best   practices, cloud   compute   and   storage   services should be deployed   and   managed independently. 
This means that potentially huge datasets need to be shipped from the storage service to the compute service to analyse the data.
This is problematic even when they are connected by a fast network, and highly exacerbated when connected across the WAN e.g. in hybrid cloud scenarios. 
To address this, minimizing the  amount  of  data  sent  across  the  network is critical to achieve good performance and low cost.
Data skipping is a technique which achieves this for SQL analytics on structured data.

Data skipping stores summary metadata for each object (or file) in a dataset. For each column in the object, the summary might include minimum and maximum values, a list or bloom filter of the appearing values, or other metadata which succinctly represents the data in that column. This metadata can then be indexed to support efficient retrieval, although since it can be orders of magnitude smaller than the data itself, this step may not be essential. The metadata can be used during query evaluation to skip over objects which have no relevant data. False positives for object relevance are acceptable since the query execution engine will ultimately filter the data at the row level. However false negatives must be avoided to ensure correctness of query results.

Unlike fully inverted database indexes, data skipping indexes are much smaller than the data itself. This property is critical in the cloud, since otherwise a full index scan could increase the amount of data sent across the network instead of reducing it. In the context of database systems, data skipping is used as an additional technique which complements classical indexes.
It is referred to as synopsis in DB2 \cite{raman2013db2} and zone maps in Oracle \cite{ziauddin2017dimensions}, where in both cases it is limited to min/max metadata.
Data skipping and the associated topic of data layout, has been addressed in recent research papers \cite{sun2014fine, shanbhag2017robust} and is also used in cloud analytics platforms \cite{sqlquery, databricks}.
Data skipping metadata is also included in specific data formats \cite{parquet, orc}.

Despite the important role of data skipping, almost all production ready implementations are limited to min/max indexes over numeric or string columns, with the exception of the ORC/Parquet formats which also support bloom filters. Moreover, queries with UDFs cannot be handled. For example, today's implementations do not support data skipping for the query below\footnote{'India' denotes a polygon with India's geospatial coordinates}. 
\begin{verbatim}
SELECT max(temp) FROM weather 
WHERE ST_CONTAINS(India, lat, lon) 
AND city LIKE '%Pur'
\end{verbatim}

We address this by implementing data skipping support for Apache Spark SQL\cite{armbrust2015spark}, and making it extensible in several ways. 
\begin{enumerate}
\itemsep0em 
\item users can define their own data skipping metadata beyond min/max values and bloom filters
\item data skipping can be applied to additional column types beyond numeric and string types e.g. images, arrays, user defined types (UDTs), without changing the source data
\item users can enable data skipping for queries with UDFs by mapping them to conditions over data skipping metadata  
\end{enumerate}
For the query above, our framework allows defining a suffix index for text columns and mapping the {\tt LIKE} predicate to exploit it for skipping, as well as mapping the {\tt ST\_CONTAINS} UDF to min/max metadata on geospatial attributes. This can reduce the amount of data scanned by orders of magnitude.
Our implementation supports plugging in metadata stores, with connectors for Parquet and Elastic Search, and  is integrated into multiple IBM products/services including IBM Cloud\textsuperscript\textregistered SQL Query, IBM Analytics Engine and IBM Cloud Pak\textsuperscript\textregistered  for Data \cite{sqlquery,iae,cp4d}.

We demonstrate various use cases for extensible data skipping, show its benefits far outweigh its costs, and show that centralized metadata storage provides significant performance benefits beyond relying on data (Parquet/ORC) formats only for data skipping. 

This paper is organised as follows. Section \ref{sec:extensible} covers extensible data skipping APIs, section \ref{sec:implementation} discusses our implementation, section \ref{sec:indexdesign} covers metadata index design, section \ref{sec:usecases} discusses experimental results, section \ref{sec:related} covers related work and section \ref{sec:conclusions} presents our conclusions.

\section{Extensible Data Skipping}
\label{sec:extensible}

Our Scala APIs allow the developer to (1) create data skipping indexes, including adding support for new index types, and (2) specify how to exploit data skipping indexes during query evaluation by mapping predicates to operations on summary metadata.
Our framework covers compositions of predicates e.g. using AND, OR and NOT, allowing expressions of arbitrary complexity.

\subsection{Extensible Data Skipping APIs}
For simplicity, we provide a running example for min/max data skipping, but our APIs can handle arbitrary predicates/UDFs and user defined metadata (e.g. {\tt LIKE/ST\_CONTAINS}
 and suffix indexes).
Useful data skipping metadata for the query below is the minimum and maximum temperature for an object (data subset\footnote{Other alternatives for data subsets are blocks, row groups etc. Our integration with Spark skips at the object level.}).
\begin{verbatim}
SELECT * FROM weather WHERE temp > 101
\end{verbatim}

\subsubsection{Index Creation}
Users can define new metadata types which extend our {\tt MetaDataType} class, such as the example below.
\begin{lstlisting}
abstract class MetadataType
case class MinMaxMetaData(col: String, 
    var min: Double, var max: Double) 
    extends MetadataType
\end{lstlisting}

Indexes are created explicitly and executed as a dedicated Spark job.
Index creation runs in 2 phases - see figure \ref{fig:indexcreationflow}. 
\begin{figure}
\centering
\includegraphics[width=3in]{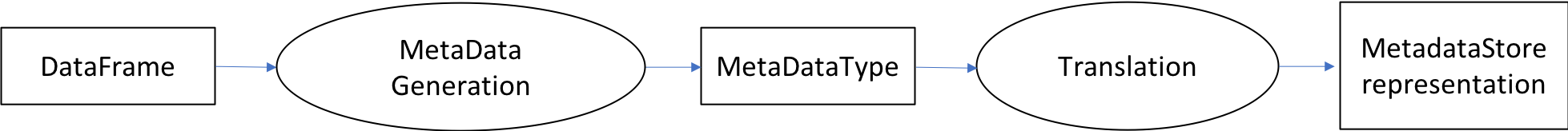}
\caption{Index creation flow}
\label{fig:indexcreationflow}
\end{figure}
The first phase accepts a Spark DataFrame (representing an object) and generates metadata having some {\tt MetaDataType}. The second phase translates this metadata to a metadata store representation.
In order to implement the first phase, the developer extends the {\tt Index} class.
\begin{lstlisting}
abstract class Index(params: Map[String, String], col: String*) {
   def collectMetaData(df: DataFrame): MetadataType
}
\end{lstlisting}
Our example {\tt MinMaxIndex} extends {\tt Index}, and {\tt collectMetadata} returns a {\tt MinMaxMetaData} instance containing minimum and maximum values for the given object column. 

\subsubsection{Query Evaluation}

Spark has an extensible query optimizer called Catalyst\cite{armbrust2015spark}, which contains a library for representing query trees and applying rules to manipulate them. We focus on query predicates i.e. boolean valued expressions typically appearing in a WHERE clause, which can be represented as Expression Trees (ETs). Figure \ref{fig:et} shows the expression tree for our example query.
\begin{figure}
\centering
\includegraphics[width=2in]{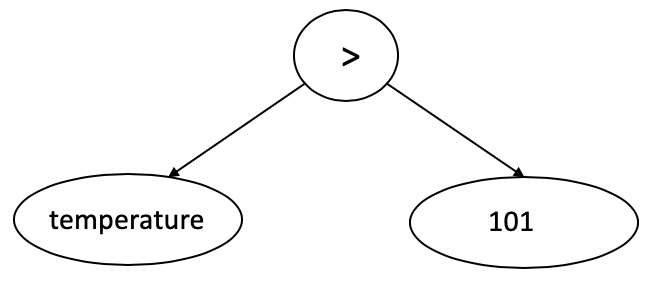}
\caption{An Expression Tree (ET) for the example query}
\label{fig:et}
\end{figure}

We analyse ETs and label tree nodes with {\em Clauses}. A Clause is a boolean condition that can be applied to a data subset $s$, typically by inspecting its metadata.
Note that for a query ET $e$, for every vertex $v$ in $e$, we denote the set of Clauses associated with $v$ by $CS(v)$.
\begin{definition}
Denote the universe of possible data subsets (i.e., $objects$) by $U$.
A Clause $c$ is a boolean function $U\to \{0,1\}$.
\end{definition}

\begin{definition}\label{representdef}
For a Clause $c$ and a (boolean) query expression $e$, we say that $c$ {\bf represents} $e$ (denoted by $c \wr e$), if for every data subset $S$, whenever there exists a row $r \in S$ that satisfies $e$, then $S$ satisfies $c$.
\end{definition}
This means that if $S$ does {\bf not} satisfy $c$, then $S$ can be safely skipped when evaluating the query expression $e$.
For example, let $e$ be
{\tt temp > 101}. Given a data subset $S$, let $c$ be the Clause $\max_{r \in S}{temp(r)} > 101$.
Then $c$ represents $e$.
Therefore, objects where \\ $\max_{r \in S}{temp(r)} <= 101$ can be safely skipped. 

Query evaluation is done in 2 phases as shown in figure \ref{fig:queryevaluationflow}. In the first phase, a query's ET $e$ is labelled using a set of clauses and the clauses are combined to provide a single clause which represents $e$. The labelling process is extensible, allowing for new index types and for new ways of using metadata. In the second phase, this clause is translated to a form that can be applied at the metadata store to filter out the set of objects which can be skipped during query evaluation.
\begin{figure}
\centering
\includegraphics[width=3in]{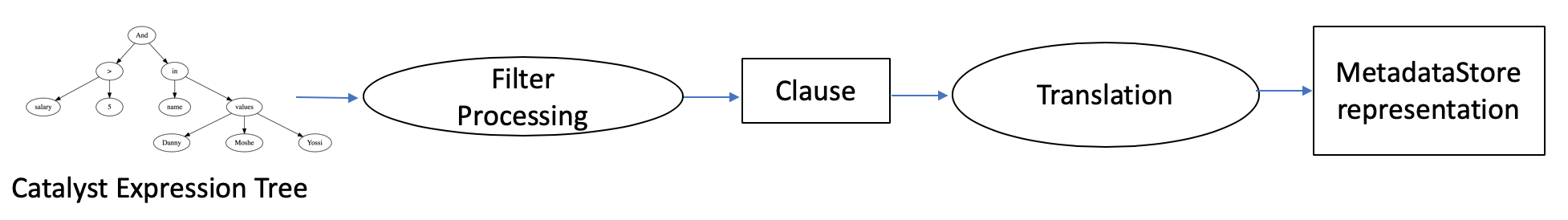}
\caption{Query evaluation flow}
\label{fig:queryevaluationflow}
\end{figure}

The labelling process is done using {\em filters}. Typically there will be one or more filters for each metadata index type. For example, we will define a {\tt MinMaxFilter} to correspond to our {\tt MinMaxIndex}. 
\begin{definition}
An algorithm $A$ is a {\bf filter} if it performs the following action:
When given an expression tree $e$ as input, for every (boolean valued) vertex $v$ in $e$, it adds a set of clauses $C$ s.t. $\forall c \in C$: $c \wr v$ to the existing set of clauses.
\footnote{Note that for a particular node, a filter might not add any clauses (this is the special case of adding the empty set).}
\end{definition}

\begin{figure}
\centering
\includegraphics[width=2.8in]{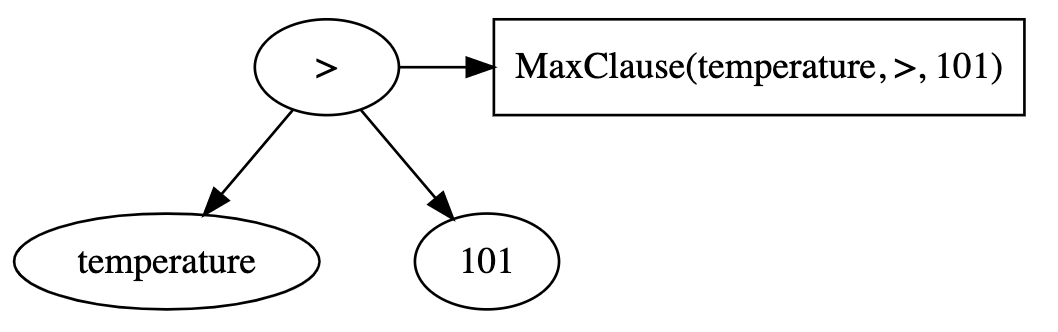}
\caption{The result of a filter on an ET}
\label{fig:filteredET}
\end{figure}

For example a filter $f$ might label our ET using {\tt MaxClause}, as shown in figure \ref{fig:filteredET}, where for a column name $c$ and a value $v$, MaxClause($c$,$>$,$v$) is defined as $\max_{r \in S}{c(r)} > v$. Since MaxClause(temperature,$>$,101) represents the node to which it was applied, $f$ acted as a filter. Since $\max_{r \in S}{c(r)}$ is stored as metadata in {\tt MinMaxMetaData}, MaxClause can be evaluated using this metadata only. 

We provide the user with APIs to define clauses and filters. A {\tt Clause} is a trait which can be extended. A {\tt Filter} needs to define the labelNode method. 
\begin{lstlisting}
case class MaxClause(col:String, op:opType, value:Literal) extends Clause
case class MaxFilter(col:String) extends BaseMetadataFilter {
 def labelNode(node:LabelledExpressionTree): Option[Clause] = {
 			node.expr match {
           case GreaterThan(attr: Attribute, v: Literal) if attr.name == col => Some(MaxClause(col, GT, v))
           case _ => None
    }}
\end{lstlisting}

Filters typically use pattern matching on the ET structure\footnote{For simplicity we left out the cases of $\leq$ and $\geq$ for {\tt MaxFilter} above.}. Similarly we can define a {\tt MinFilter} which can label a tree with {\tt MinClause}s. 
Patterns can also match against UDFs in expression trees e.g. {\tt ST\_CONTAINS} - see also section \ref{sec:geoudfs} for queries using UDFs. 

In some cases a filter's patterns may need to match against complex predicates using AND/OR/NOT. For example, the GeoBox index (section \ref{sec:indexdesign}) 
stores a 2 dimensional bounding box for each object and 
the corresponding filter needs to match against an AND with child constraints on both lat and lng. 
Figure \ref{fig:filterGeobox} illustrates this. 

\begin{figure}
\centering
\includegraphics[width=2.8in]{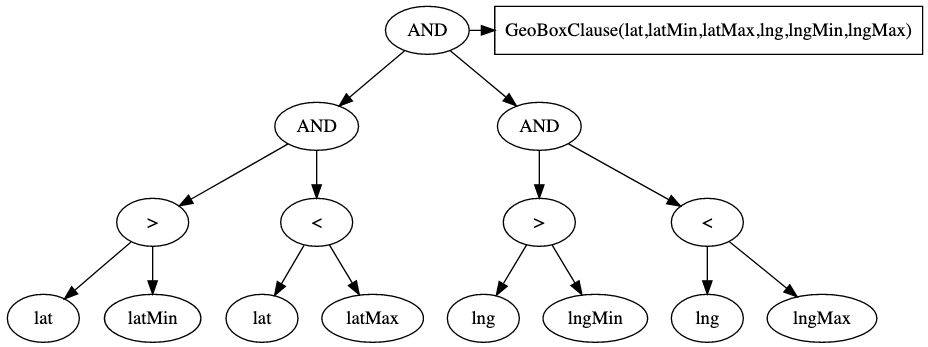}
\caption{The result of a geobox filter on a complex query expression}
\label{fig:filterGeobox}
\end{figure}

Each {\tt MetaDataFilter} needs to be registered in our system, and during query optimization we inspect the types of metadata that were collected and run the relevant filters on the query's ET. Running the complete set of registered filters will generate an ET where each node can be labelled by multiple Clauses.
For every vertex $v$ in $e$, we denote the set of Clauses associated with $v$ by $CS(v)$.
We recursively merge all of an ET's Clauses to form a unified Clause which represents it. 
This Clause is then applied to the metadata to make a final skipping decision. 
For a full formal description of the algorithm used and proof of correctness, see  Appendix~\ref{appendix:correctness}

\section{Implementation}
\label{sec:implementation} 
We implemented data skipping support for Apache Spark SQL\cite{armbrust2015spark}
as an add-on Scala library which can be added to the classpath and used in Spark applications. Our work applies to storage systems which implement the Hadoop FileSystem API, which includes various object storage systems as well as HDFS. We tested our work using IBM Cloud Object Storage (COS) and the Stocator connector \cite{stocator,vernik2018stocator}.
Metadata is stored via a pluggable API which we describe in section \ref{sec:metadatastore}.
The library supports multiple levels of extensibility: code which implements any of our extensible APIs such as metadata types, and clause and filter definitions, as well as additional metadata stores, can be added as plugin libraries.

\subsection{Spark Integration}
Spark uses a partition pruning technique to filter the list of objects to be read if the dataset is appropriately partitioned.
Our approach further prunes this list according to data skipping metadata  
as shown in figure~\ref{fig:ourflow}.

Our technique applies to all Spark supported native formats e.g. JSON, CSV, Avro, Parquet, ORC, and can benefit from optimizations built into those formats in Spark.
Unlike approaches which embed data skipping metadata inside the data format which require reading footers of every object, our approach avoids touching irrelevant objects altogether.
It also avoids wasteful resource allocation because when relying on a format's data skipping, Spark allocates resources to handle entire objects, even when only object footers need to be processed. 
\begin{figure}
\centering
\includegraphics[width=3in]{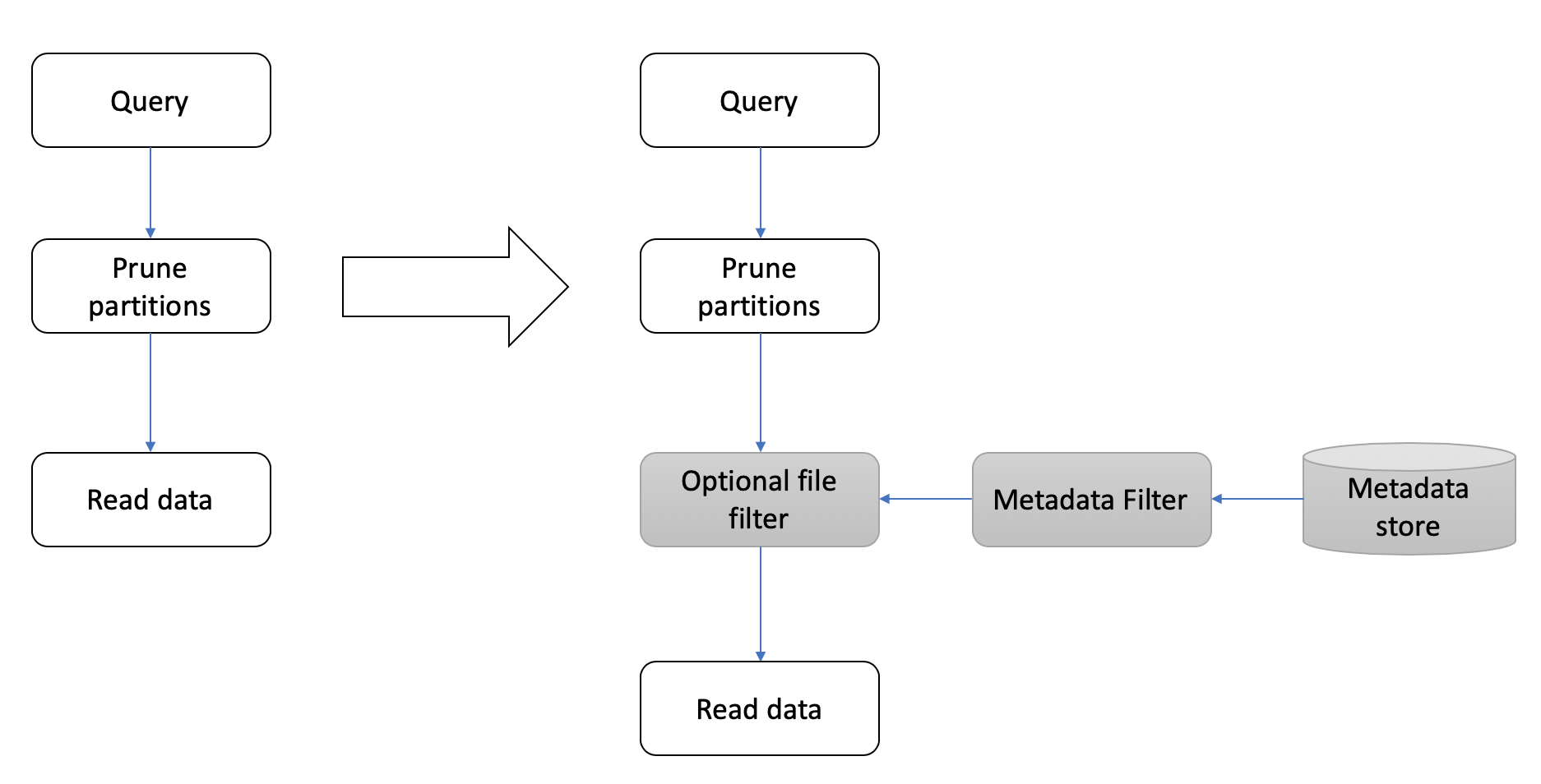}
\caption{Modified Spark query execution flow after integration with extensible data skipping}
\label{fig:ourflow}
\end{figure}
We provide an API for users to retrieve how much data was skipped for each query. 

We used APIs provided by Spark's Catalyst query optimizer to achieve this without changing core Spark. In particular, we added a new optimization rule using the Spark session extensions API\cite{sunithablog}. Spark SQL maintains an \\ {\tt InMemoryFileIndex} which tracks the objects to be read for the current query and their properties. Our rule wraps the {\tt InMemoryFileIndex} with a new class extending it by adding the additional filtering step from figure~\ref{fig:ourflow}. 

We refrain from skipping objects when our metadata about them is stale. This can happen if objects are added, deleted or overwritten in a dataset after indexing it. We keep track of freshness using last modified timestamps, which are retrieved during file listing by the {\tt InMemoryFileIndex}.  We also provide a refresh operation, which updates stale metadata.

\subsection{Metadata Stores}
\label{sec:metadatastore}
We support a pluggable API for metadata stores including the specification of how metadata and clauses should be translated for a particular store. 
This includes the indexing time translation API for figure \ref{fig:indexcreationflow} and the query time translation API for figure \ref{fig:queryevaluationflow}.
The key property is that these translations should preserve the correctness of our skipping algorithm. 
We used this API to implement both Parquet and Elastic Search\cite{elasticsearch} metadata stores.

It is now widely accepted practice to use the same storage system for both data and metadata\cite{delta,iceberg,hudi}, avoiding deployment of an additional metadata service.
This is achieved using our Parquet metadata store, and all storage systems implementing the Hadoop FS API are supported.
Relevant metadata indexes are scanned prior to query execution, but this cost is not significant, since metadata is typically considerably smaller than data. 
By leveraging Parquet's column-wise compression and projection pushdown for metadata, we minimize the amount of metadata that needs to be read per query, ensuring low overhead. 
Use of Parquet also allows generating and storing metadata for multiple columns together, resulting in better indexing and refresh performance, compared with storing indexes on each column separately.

\subsection{Protecting Sensitive Data and Metadata}
Security and privacy protection for sensitive data are essential for today's cloud services.
Parquet supports column-wise encryption of sensitive columns in a modular and efficient fashion\cite{parquetspec,parqencperf}, and being format-agnostic, our library supports skipping over encrypted parquet data transparently.
However, an end-to-end solution needs to encrypt metadata, since it can also leak sensitive information. 
To prevent leakage, when storing metadata in Parquet, we implemented an option to encrypt indexes on sensitive columns, by assigning a key to each index. 
A user can choose the same key used to encrypt the column the index originates from, choose another key, or leave the index as plaintext.
This scheme enables scenarios such as storing data and metadata at a shared location, where each user can only access a subset of the columns and indexes according to their keys. 

Spark's partition pruning capability relies on either (1) a widely accepted naming convention which names appropriately partitioned data objects according to their partitioning column name and column value or (2) a Hive metastore. The first option leaks metadata into object names, and therefore partitioning according to sensitive columns is problematic. 
Use of a Hive metastore in a multi-tenant cloud service pushes the problem of managing sensitive multi-tenant metadata to the underlying database. 
An alternative is to rely on our data skipping framework for partition pruning, thereby ensuring end-to-end data and metadata protection, without sacrificing performance.

\section{Metadata Index Design}
\label{sec:indexdesign}
In this section we explain the requirements of a good metadata index type and cover indicators of skipping effectiveness. 
We show that in theory both selecting and designing optimal indexes are hard problems. 
However, we demonstrate practical choices that work well in this and the following section.  
We survey various index types implemented using our APIs with a summary in table~\ref{table:indextypes}.  

Our goal is to minimize the total number of bytes scanned, because there is a close correlation between this and query completion time (e.g. see section~\ref{sec:usecases}). Moreover, users of serverless SQL services are typically billed in proportion to the number of bytes scanned\cite{athena,sqlquerypricing}.

For each query, prior to reading the data, the relevant metadata is scanned and analyzed. As long as the metadata is much smaller than the data, this approach can significantly reduce the amount of data scanned overall. 
For big datasets the overhead of scanning metadata is usually insignificant compared to the benefits of skipping data (see figure~\ref{fig:breakdown}), and in some cases metadata can also be cached in memory or on SSDs. When using our Parquet metadata store, we read only the relevant metadata indexes by using Spark and Parquet column projection capabilities. 

\subsection{Indicators of Skipping Effectiveness}\label{skippingindicators}
Given a dataset (set of rows) $D$ and a query $Q$, a row $r$ in $D$ is relevant to $Q$ if $r$ must be read in order to compute $Q$ on $D$. Let $D_r$ denote the set of relevant rows in $D$. Assuming $D$ is stored as objects\footnote{alternatively other units can be considered such as blocks, row groups etc.}, let $O$ denote the set of all  objects for $D$, let $O_r$ denote the set of objects relevant to $Q$ (i.e. having at least one relevant row), and let $O_m$ be the set of objects deemed relevant according to the metadata associated with $D$. Note that $O_r \subseteq O_m$. Note that $O_s = O - O_m$ is the set of objects that can be skipped. 

Denote the number of rows in object $o$ (or dataset $D$) as $|o|$ ($|D|$). All definitions below are w.r.t. a dataset $D$ and a query $Q$.
\begin{definition}
The {\bf selectivity} $\sigma$ of a query is the proportion of relevant rows $\sigma = \frac{|D_r|}{|D|}$  
\end{definition}
Data skipping can potentially reduce bytes scanned for selective\footnote{Selectivity ranges between 0 and 1.``Highly selective" queries have close to 0 selectivity} queries. 
The definitions use relevant rows rather than rows in the result set to account for queries which perform further computations such as aggregation. 
\begin{definition}
The {\bf layout factor} $\lambda$ of a query is the proportion of relevant rows in relevant objects\\
$\lambda = \frac{|D_r|}{\sum_{o \in O_r}|o|}$
\end{definition}
Mixing relevant and irrelevant rows in the same object decreases the layout factor. A high layout factor (grouping relevant rows together) increases the potential for data skipping. To realise this potential we need effective metadata.
\begin{definition}
The {\bf metadata factor} $\mu$ of a query is \\ 
$\mu = \frac{\sum_{o \in O_r}|o|}{\sum_{o \in O_m}|o|}$
\end{definition}
The metadata factor is closely related to the metadata's false positive ratio - a low false positive ratio gives rise to a high metadata factor. In addition the metadata factor takes into account the relative size of each object. A high metadata factor denotes that the metadata is close to optimal given the data layout. 
\begin{definition}
The {\bf scanning factor} $\psi$ of a query is the proportion of rows actually scanned (using metadata) \\
$\psi = \frac{\sum_{o \in O_m}|o|}{|D|}$
\end{definition}
Our aim is to achieve the lowest possible scanning factor. According to our definitions
\begin{equation}
\label{eqn:indicators}
\psi = \frac{\sigma}{\lambda\mu}
\end{equation}
To achieve this for a selective query we need $\lambda\mu$ to be high, and we are equally dependent on good layout and effective metadata\footnote{Note that the scanning factor is not defined for queries with 0 selectivity}.

We focus here on metadata effectiveness for any given data layout, and refer the reader to previous work regarding data layout optimization\cite{sun2014fine, shanbhag2017robust}. 
In practice, often data layout is given and cannot be changed e.g. legacy requirements, compliance, encryption of one or more sensitive columns. In other cases, re-layout of the data is too costly, or it might be difficult to meet the needs of multiple conflicting workloads without duplicating the entire dataset.

Our approach is to enable an extensible range of metadata types, which cater to data within a reasonable range of layout factors. 
Generating data skipping metadata is typically significantly cheaper than changing the data layout, since no shuffling of the data is needed. Moreover, unlike data layout, it can be done without write access to the dataset and only requires read access to the column(s) at hand. Each user can potentially store metadata corresponding to their particular workload.

On the other hand, using equation \ref{eqn:indicators}, we can identify cases where the layout factor is prohibitively low and good skipping is unachievable without re-layout.

To take averages of skipping indicators when considering multiple queries, we use the geometric mean, following\cite{sel-paper}. Let $\gm[X]$ denote $(\prod_{i=1}^n x_i)^{\frac{1}{n}}$. 
Given a dataset $D$ and a workload with queries $q_1,\ldots,q_n$, where for each $q_i$ we have $\psi_i = \frac{\sigma_i}{\lambda_i \mu_i}$, then we also have
\begin{equation}\label{eqn:factorsequation}
\gm[\psi] = \frac{\gm[\sigma]}{\gm[\lambda] \gm[\mu]}
\end{equation}
We apply this approach to measuring the skipping indicators on real world datasets and workloads in section \ref{sec:usecases}.

\subsection{The Index Selection Optimization Problem}
Given a dataset $D$ and query workload (set of queries) $Q$, it is natural to ask what is the optimal set of metadata indexes we can store to achieve the lowest possible scanning factor. Since the workload and data layout are given, $\sigma$ and $\lambda$ are given, and to achieve low $\psi$ we need to achieve high $\mu$. We assume that every metadata index $i$ has a cost $c_i$, and that we need to stay within a given metadata budget $K$. A natural cost definition is the size of the metadata in object storage. 
We also assume that each index $i \in I$ provides a benefit $v_i$ which in our case corresponds to the increase in $\mu$ as a result of $i$. 
Ideally, given $K$ and a set of candidate metadata indexes $I$, one could choose an optimal subset $I' \subseteq I$ which gives maximal $\mu$ while staying within budget. 
We show that this problem is NP-hard using a reduction from the knapsack problem. Previous work showed that the problem of finding a data layout providing optimal skipping is also NP-hard\cite{sun2014fine}.

\begin{problem}
\label{metadata-optimization-problem}
Given dataset $D$, workload $Q$, a set of indexes $I$, and a metadata budget $K$, find $I' \subseteq I$ that maximizes $\sum_{i \in I'}v_i$ subject to $\sum_{i \in I'}c_i \leq K$.
\end{problem}

\begin{claim}
Problem \ref{metadata-optimization-problem} is NP-hard.
\end{claim}
\begin{proof}
By reduction from \{0,1\}-knapsack. Knapsack item weight and value correspond to the cost and benefit of an index respectively, and knapsack capacity corresponds to the metadata budget. Clearly, maximizing the value of items in the knapsack within capacity is equivalent to maximizing index benefit within a metadata budget.
\end{proof}
\begin{remark}
This formulation shows that even in the special case where the benefit of an index is independent from other indexes, the problem is hard. In the general case, the benefit of indexes is relative since, for example, an index which achieves maximal $\mu$ renders further addition of indexes obsolete. 
\end{remark}

Given a fixed metadata budget, choosing optimal indexes is a hard problem.  
However, for many index types\footnote{all index types in table~\ref{table:indextypes} except for value list and prefix/suffix indexes} we store a fixed \#bytes per object, thereby bounding the index size to a small fraction of the data size. Using such index types, it is reasonable to index all data columns, assuming the metadata is stored in the same storage system as the data (i.e. with the same storage/access cost).

\subsection{An Index Design Optimization Problem}
Choosing an optimal set of indexes is hard. What about designing a single optimal index?
We show that this is hard even for a range query workload on a single column.  
Consider a single column $c$ with a linear order e.g. integers, and a workload $Q$ with {\bf range queries} over $c$ i.e. queries of the form 
\begin{verbatim}
SELECT * FROM D 
WHERE c between c1 and c2
\end{verbatim}
Storing min/max metadata only for $c$ may not achieve maximal $\mu$, for example, when an object's rows have gaps in column $c$ between the min and max values. In this case if $c1$ and $c2$ both fit inside the gap then min/max metadata will give a false positive for the query above. 
A gap list metadata index could store a list of such gaps per object, and be used to skip objects having gaps covering the intervals used in queries. Given a dataset $D$, a workload $Q$ and metadata budget of $k$ gaps, which gaps should be stored to give optimal $\mu$? (We assume the cost of each gap is equal). An algorithm which achieves this is provided in \cite{eagleEyed}.
However, we show that allowing queries with disjunction turns this into a hard problem. 
\begin{problem}
\label{index-optimization-problem}
Given a dataset $D$ with a column $c$ having a linear order, a workload $Q$ comprising of disjunctive range queries over $c$, and a metadata budget of $K$ gaps, find a set of $k$ gaps where $k \leq K$ such that $\mu$ is maximized.
\end{problem}

\begin{problem}
(Densest k-Subhypergraph problem) Given a hypergraph $G= (V, E)$ and a parameter $k$, find a set of $k$ vertices with maximum number of hyperedges in the subgraph induced by this set\cite{ChlamtacDKKR16}.
\end{problem}
\begin{claim}
Problem \ref{index-optimization-problem} is NP-hard.
\end{claim}
\begin{proof}
By reduction from the densest k-Subhypergraph problem. Given  $G= (V, E)$ and $k$, we construct an input to problem \ref{index-optimization-problem} as follows. We create a dataset with one object $O$ such that its column $c$ induces $|V|$ gaps - $\{g_1, g_2, .., g_{|V|}\}$, and use the function $f(v_i) = g_i$ to map each vertex to a gap. 
Each hyperedge $e\in E$ is mapped to a query with a WHERE clause comprised of a predicate of the form:
$\vee_{v\in e}{c \in {f(v)}}$. In order to skip $O$ for this query we need exactly those gaps in $\{ f(v) | v \in e \}$.
In this setting maximizing $\mu$ (the number of queries where $O$ is skipped) is equivalent to finding the densest k-Subhypergraph.
\end{proof}

\subsection{Metadata Index Types}

Table~\ref{table:indextypes} contains a summary of common index types (MinMax, BloomFilter) as well as novel ones we found useful for our use cases and implemented for our Parquet metadata store. 
All metadata enjoys Parquet columnar compression and efficient encoding - therefore the Bytes/Object values in the table can be considered an upper bound.

\begin{table*}
\centering
\begin{threeparttable}
\caption{Data Skipping Index Types}
\begin{tabular}{ | c | l | c | c | c | }\hline
{\bf Index Type}& {\bf Description} & {\bf Column Types} &{\bf Handles Predicates}\tnote{1} & {\bf Bytes/Object}\tnote{2}\\ \hline
MinMax &  Stores minimum/maximum values for a column   & ordered & $p(n,c)$ & $2b$ \\ \hline
GapList & Stores a set of $k$ gaps indicating ranges where there are & ordered & $p(n,c)$ & $kb$ \\ 
& no data points in an object & & & \\ \hline
GeoBox & Applies to geospatial column types e.g. Polygon, Point. & geospatial & geo UDFs& $2xb$ \\ 
&  Stores a set of $x$ bounding boxes covering data points& & & \\ \hline
BloomFilter & Bloom filter is a well known technique\cite{bloom} & hashable & $n=c, n \in C$ & $\frac{-v\ln{f}}{\ln^2{2}}$ {\tiny(in bits)} \\ \hline
ValueList & Stores the list of unique values for the column & has =, text & $n=c, n \in C$, LIKE & $vb$ \\ \hline
Prefix & Stores a list of the unique prefixes having $b_1$ characters & text & LIKE 'pattern\%'  & $v_{1}b_1$ \\ \hline
Suffix & Stores a list of the unique suffixes having $b_2$ characters & text & LIKE '\%pattern'  & $v_{2}b_2$ \\ \hline
Formatted & Handles formatted strings. There are many uses cases. & text & template based UDFs & varies \\ \hline
MetricDist & Stores an origin, max and min distance per object  & has metric dist& metric distance UDFs & $2m + b$  \\ \hline
\end{tabular}
\label{table:indextypes}
\begin{tablenotes}
\item [1]$p \in \{<,\leq, >, \geq,=\}$. $n$ is a column name and $c$ is a literal, $C$ is a set of literals. 
\item [2]$b$ is the (average) number of bytes needed to store a single column element. $v$ is the number of distinct values in a column for the given object. $k$ is the number of gaps (configurable). $x$ is the number of boxes per object. $v_1$ ($v_2$) is the number of distinct values with prefix (suffix) of size $b_1$ ($b_2$). $m$ is the number of bytes needed to store a distance value. $f$ is the false positive rate ($f \in(0,1)$).
\end{tablenotes}
\end{threeparttable}
\end{table*}

The {\bf MetricDist} index enables similarity search queries using UDFs based on any metric distance e.g. Euclidean, Manhattan, Levenshtein.  Applications include document and genetic similarity queries. Recently semantic similarity queries have been applied to databases\cite{cognitivedb}, where values are considered similar based on their context, 
allowing queries such as ``which employee is most similar to Mary?". Assuming a metric function for similarity, extensible data skipping can be successfully applied.

Additional index types can be easily integrated by implementing our APIs - example candidates include SuRF\cite{surf}, HOT\cite{hot}, HTM\cite{htm}. 
Recent work demonstrated the use of range sets (similar to our gap lists) to optimize queries with JOINs\cite{dips}. 
Adding a new index type via our APIs requires roughly 30 lines of new code.

\subsection{A Hybrid Index}
When a column typically has low cardinality per object, a value list is both more space efficient than a bloom filter and avoids false positives. However, for high cardinality, value list metadata size can approach that of the data. In order to achieve the best of both worlds, we implemented a hybrid index, which uses a value list up to a certain cardinality threshold, and a bloom filter otherwise. 
We now explain how we determined an appropriate threshold. 

Assuming equality predicates only, we compare value list and bloom filter indexes using the formulas presented in table \ref{table:indextypes}.
Our aim is to minimize the total bytes scanned for data and metadata.
Given an object of size $|o|$, a column with $v$ distinct values each one of size $\bar{b}$ bits, and a workload $Q = {\{q_i\}}_{i=1}^{n}$ of exact match queries, let $E_i \in \{0,1\}$ be the event in which $o$ must be read for $q_i$. It follows that the average data to be scanned for the workload using a bloom filter index is approximately $ \frac{1}{n} \sum_{i=1}^{n} (\frac{-v\ln{f}}{\ln^2{2}} + E_i|o| + (1-E_i)f|o|)$. The average data to be scanned for the workload using value list is exactly $ \frac{1}{n} \sum_{i=1}^{n}(v\bar{b} + E_i|o|)$.
Therefore, a value list index is preferable when:
\[v(\bar{b} + \frac{\ln{f}}{\ln^2{2}}) < f|o|(1 - \frac{1}{n} \sum_{i=1}^{n}E_i) \]
The term $\frac{1}{n} \sum_{i=1}^{n}E_i $ can be approximated using the expected scanning factor when using a value list index, which can be derived from the workload mean layout and selectivity factors using equation \ref{eqn:factorsequation}.

For example, given an object of size 64MB with a string column of up to 64 characters ($\bar{b}=512$) and a target scanning factor of 0.01, a value list up to 10,088 elements is preferable over a bloom filter with $f = 0.01$.
We implemented a {\bf hybrid index} which creates a bloom filter or value list per object according to the column cardinality. 
By default we use a threshold of 10K elements based on the above example, but this threshold can be changed according to dataset properties.

\section{Experimental Results}
\label{sec:usecases}
We focus on use cases where data is born in the cloud at a high, often accelerating, rate so highly scalable and low cost solutions are critical.  
We demonstrate our library for geospatial analytics (representing IoT workloads in general) and log analytics on 3 proprietary datasets. 
We collect skipping effectiveness indicators and discuss their effect on the scanning factor (hence data scanned). All experiments were conducted using Spark 2.3.2 on a 3 node IBM Analytics Engine cluster, each with 128GB of RAM, 32 vCPU, except where mentioned otherwise.
The datasets are stored in IBM COS. All experiments are run with cold caches.

A proprietary (1) {\bf Weather Dataset} contains a 4K grid of hourly weather measurements. The data consists of a single table with 33 columns such as latitude, longitude, temperature and wind speed. The data was geospatially partitioned using a KD-Tree partitioner\cite{shanbhag2017robust}.
One month of weather data was stored in 8192 Parquet objects using snappy compression with a total size of 191GB.

The two proprietary http server log datasets below are samples of much larger datasets and use Parquet with snappy compression:

 (2) A \textbf{Cloud Database Logs} dataset, consisting of a single table with 62 columns such as db\_name, account\_name, http\_request. The data was partitioned daily with layout according to the account\_ name for each day, resulting in 4K objects with a total size of 682GB.
 
(3) A {\bf Cloud Storage Logs} dataset, consisting of a single table with 99 columns such as container\_name, account\_name, user\_agent. The data was partitioned hourly, resulting in 46K objects with a total size of 2.47TB.

\subsection{Indexing}
Use of our APIs allows adding new index types achieving similar performance to native index types with little programmer effort. 
Table \ref{table:indexstats} in appendix~\ref{sec:indexstats} reports statistics for indexing a single column using various index types on our datasets. 
In addition, we implemented an optimization which reads min/max statistics from Parquet footers, which gives significant speedups when only MinMax indexes are used on Parquet data\footnote{If additional index types are used it provides no benefit since the Parquet row groups need to be accessed in any case.}.

Figure~\ref{fig:indexingcols} shows that indexing multiple columns using the Hybrid index is significantly faster than indexing each column separately\footnote{other index types behave similarly}, even for Parquet data where columns are scanned individually. 
For MinMax the indexing time remains low (benefiting from our MinMax optimization) and flat when varying the number of columns.

\begin{figure}
    \subfloat{\includegraphics[width=0.24\textwidth]{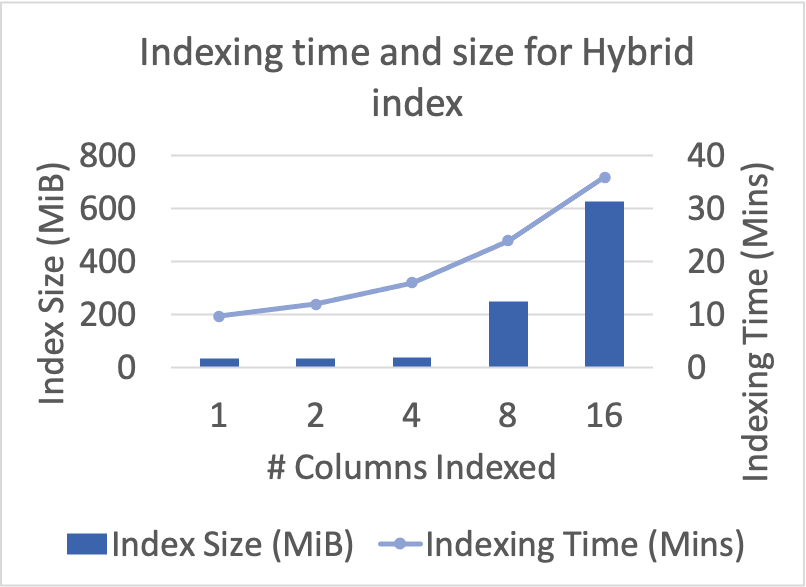}}
    \subfloat{\includegraphics[width=0.24\textwidth]{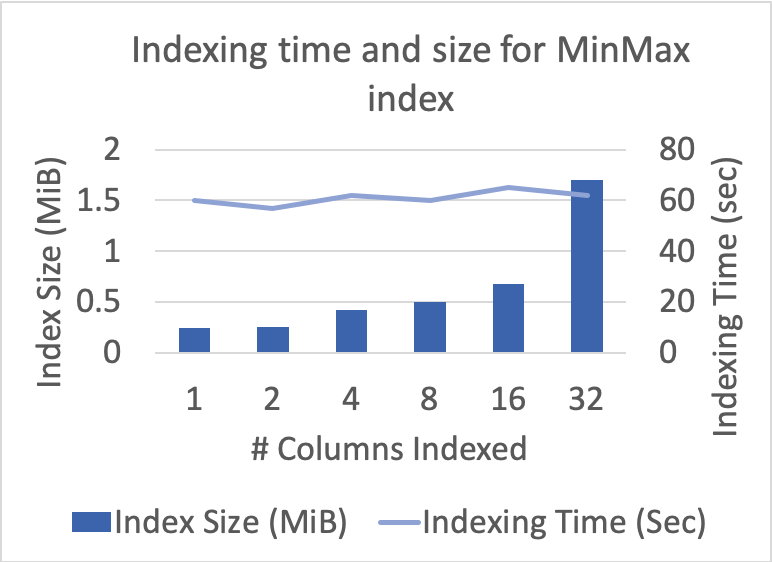}}
    \caption{Indexing time/size vs \#columns (log scale) with Hybrid (cloud database logs) and MinMax (weather)}
    \label{fig:indexingcols}
\end{figure}

We note that indexing can be done per object at data generation or ingestion time, and can alternatively be done using highly scalable serverless cloud frameworks e.g. \cite{pywren}.

\subsection{Metadata versus Data Processing}

\begin{figure*}
    \subfloat[]{\includegraphics[width=0.48\textwidth]{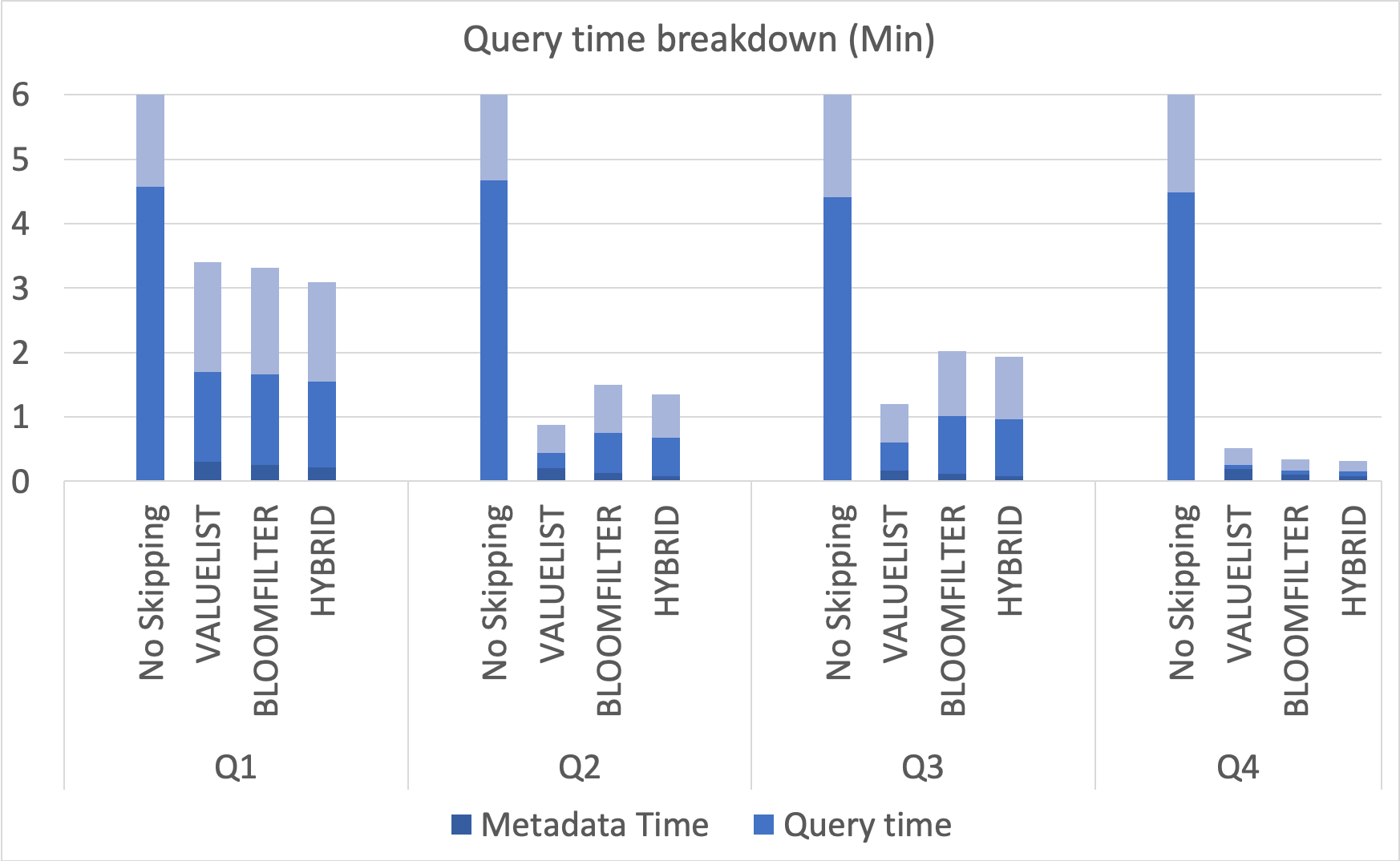}} 
    \quad
    \subfloat[]{\includegraphics[width=0.48\textwidth]{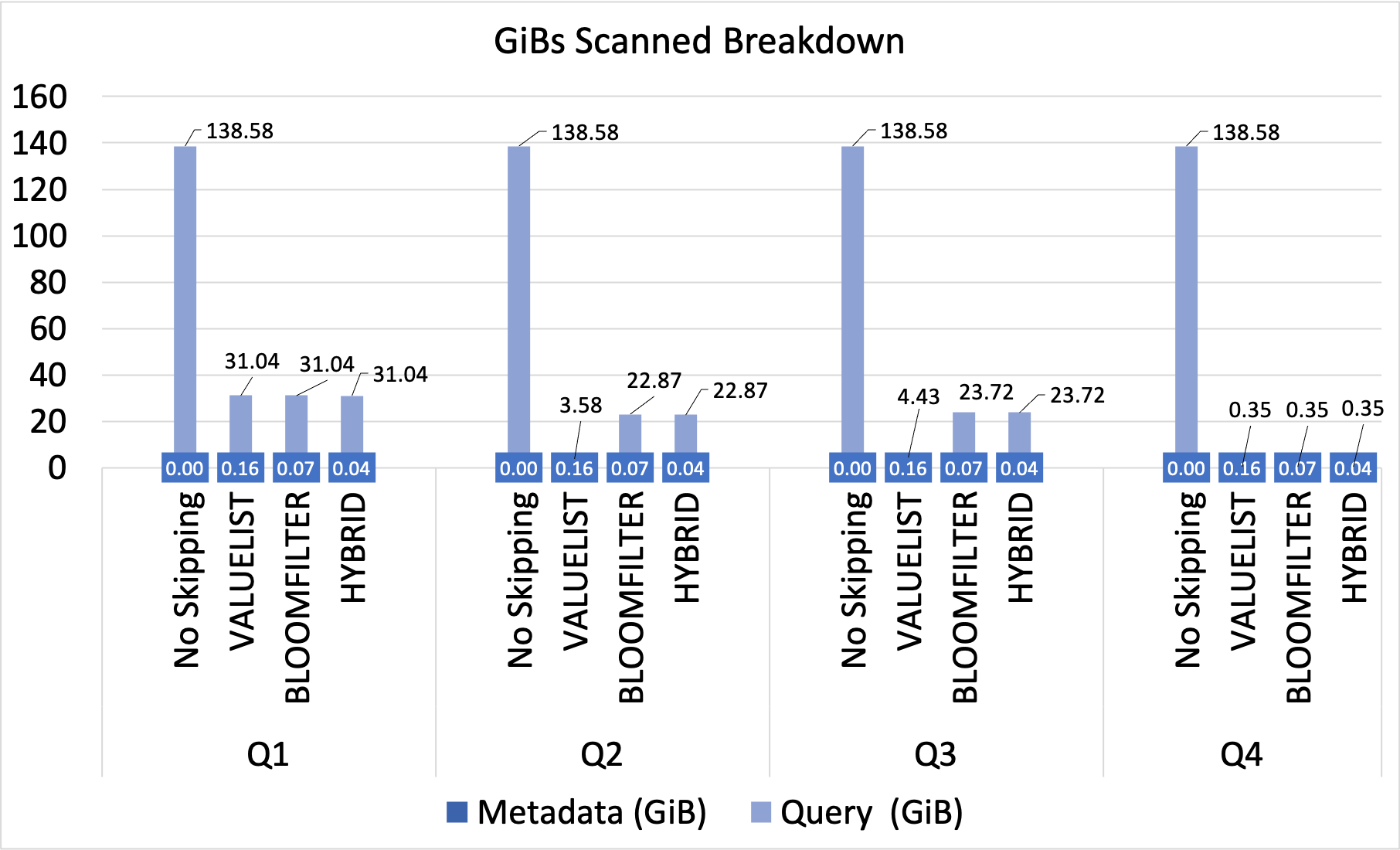}} 
    \caption{Breakdown of time spent on data and metadata processing for 4 queries on cloud database logs and corresponding bytes scanned}
    \label{fig:breakdown}
\end{figure*}

Figure~\ref{fig:breakdown} shows time and bytes scanned for 4 queries searching for different values of the db\_name column (cloud database logs dataset). The queries retrieve 8 columns, and we compare ValueList, BloomFilter and Hybrid indexes on the db\_name column, and in all cases either ValueList or the Hybrid index outperforms BloomFilter (whereas BloomFilter is the index most widely adopted in practice).  There is a clear correspondence between bytes scanned and query completion times, and data skipping reduces query times roughly between x3 and x20. In all cases, the time spent on metadata processing is a small fraction of the overall time. For all queries, the Hybrid index requires the least metadata processing time because of its smaller size. For Q4, when almost all data is skipped, the Hybrid index is superior for this reason. For Q2 and Q3, Hybrid and BloomFilter incur false positives and so retrieve more data than ValueList, resulting in longer query times. 

We point out that for this scenario it only takes around 3 queries to save the 10 mins that were spent on indexing the db\_name column.
On the other hand, the overhead for all queries (selective and non selective) with all indexes is less than 20 seconds per query. 
In terms of bytes scanned, we scanned 6.73 GB to index the db\_name column, whereas each query saves over 100GB (because of the additional columns retrieved). 
Therefore in terms of cost, a user can achieve payback after a single query.

\subsection{Data Skipping for Geospatial UDFs}
\label{sec:geoudfs}
We demonstrate data skipping for queries with predicates containing UDFs. To our knowledge, no other SQL engine supports this, since query optimizers typically know very little about UDFs. 
We used our extensible APIs to create filters that identify UDFs from IBM's geospatial toolkit\cite{geospatialtoolkit} and map them to MinMax and GeoBox index types. 
Supported predicates include containment, intersection, distance and many more \cite{geospatialtoolkitfunctions}.

For example, the following query retrieves all data whose location is in the Bermuda Triangle.
Without data skipping, the entire dataset needs to be scanned. 
\begin{verbatim}
SELECT * FROM weather WHERE 
ST_CONTAINS(ST_WKTToSQL(
'POLYGON((-64.73 32.31,...))'), 
ST_POINT(lat, lng))
\end{verbatim}
In order to support skipping we can either use the GeoBox index on the pair of lat/lng columns, or use independent MinMax indexes on both lat and lng. For each case we map the relevant UDFs to the corresponding Clauses. The GeoBox index has the advantage that it can handle lower layout factors by using multiple boxes per object. Since we partitioned the dataset according to lat/lng, 
the MinMax approach is also effective.

Figure \ref{fig:geospatialskippingvsnoskipping} compares running ST\_Contains queries with and without data skipping.\footnote{The results for ST\_Distance are similar.} The queries were run on an extrapolation of the weather dataset to a 5 year period. We used MinMax indexes resulting in 11MB of metadata for close to 12TB of data. 
The specific query we ran has the same form as our example query, and selects data with location in the Research Triangle area of North Carolina, with time windows ranging between 1 to 12 months. We achieved a cost and performance gap which is over 2 orders of magnitude - the gap increases in proportion to the size of the time window. For a 5 month window we achieved a x240 speedup. 
The cost gaps reflected by amount of data scanned are similar. We conclude that even with a high layout factor, running queries with UDFs directly on big datasets is clearly not feasible without extensible data skipping.

\begin{figure}
    \subfloat{\includegraphics[width=0.24\textwidth]{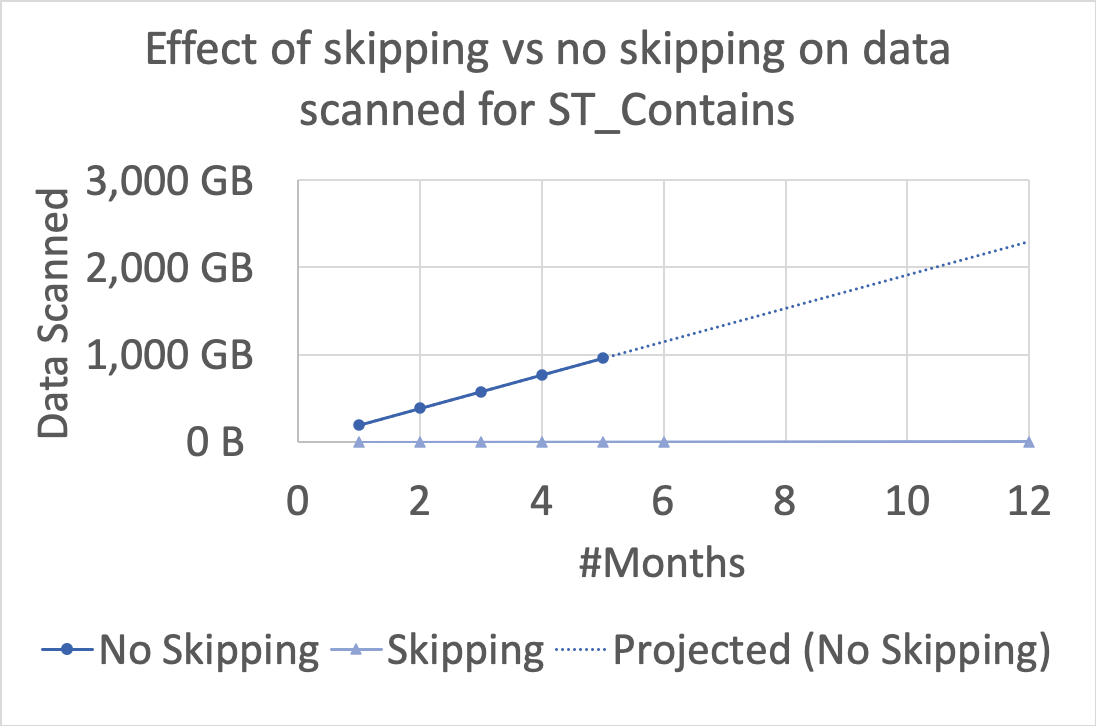}}
    \subfloat{\includegraphics[width=0.24\textwidth]{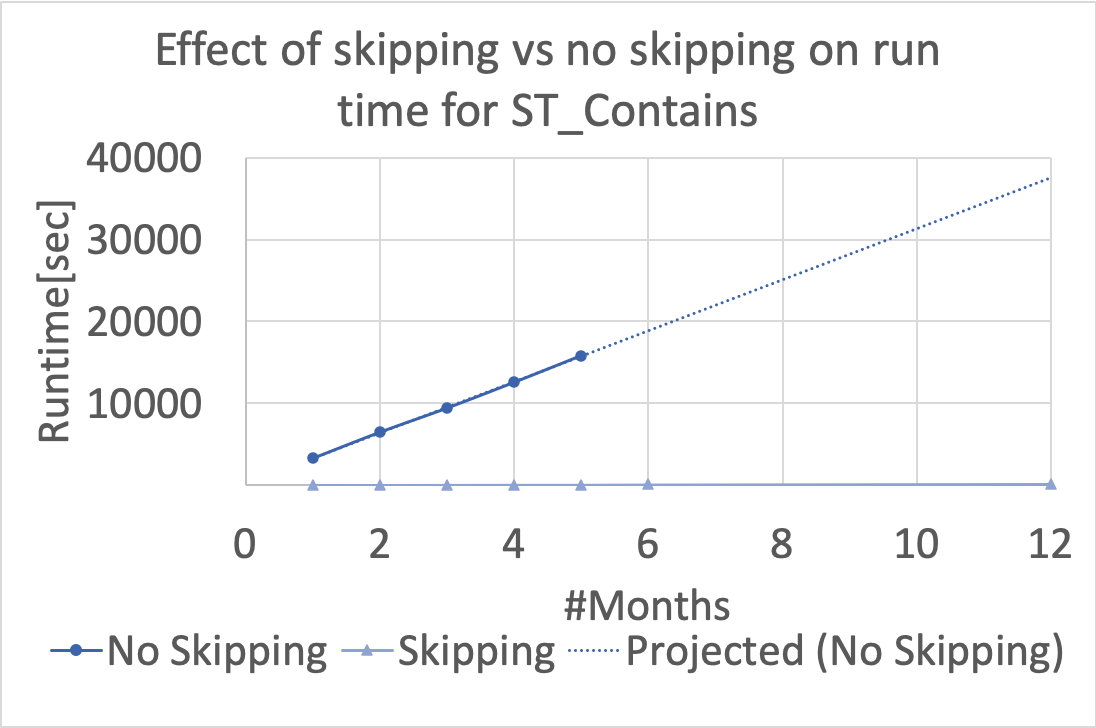}}
    \caption{Effects of skipping versus no skipping for ST\_Contains applied to the weather dataset with varying time window sizes}
    \label{fig:geospatialskippingvsnoskipping}
\end{figure}

\subsection{Benefits of Centralized Metadata}
\label{sec:centralized}
An alternative approach is to apply geospatial data layout and rewrite queries to exploit min/max metadata, if available in the storage format. This approach requires users to rewrite queries manually, or else query rewrite needs to be implemented for each query template. For example, the previous query could be rewritten to the one below
\begin{verbatim}
SELECT * FROM weather WHERE 
ST_CONTAINS(ST_WKTToSQL(
'POLYGON((-64.73 32.31,...))'), 
ST_POINT(lat, lng)) 
AND lat BETWEEN 18.43 AND 32.31 
AND lng BETWEEN -80.19 AND -64.73
\end{verbatim}

Our approach uses centralized metadata which avoids reading the footers of irrelevant Parquet/ORC objects altogether. This achieves a performance boost for 2 main reasons: overheads for each GET requests are relatively high for object storage, and Spark cluster resources are used more uniformly and effectively. The bytes scanned are reduced both by avoiding reading irrelevant footers and by metadata compression, which lowers cost.
Figure \ref{fig:geospatialskippingvsrewrite} compares the cost and performance of extensible data skipping to a query rewrite approach.
Since the data is partitioned geospatially, both identify the same objects as irrelevant.
However, 
our centralized metadata approach performs x3.6 better at run time at x1.6 lower cost for 5 year time windows, demonstrating significant benefit. 

\begin{figure}
    \subfloat{\includegraphics[width=0.24\textwidth]{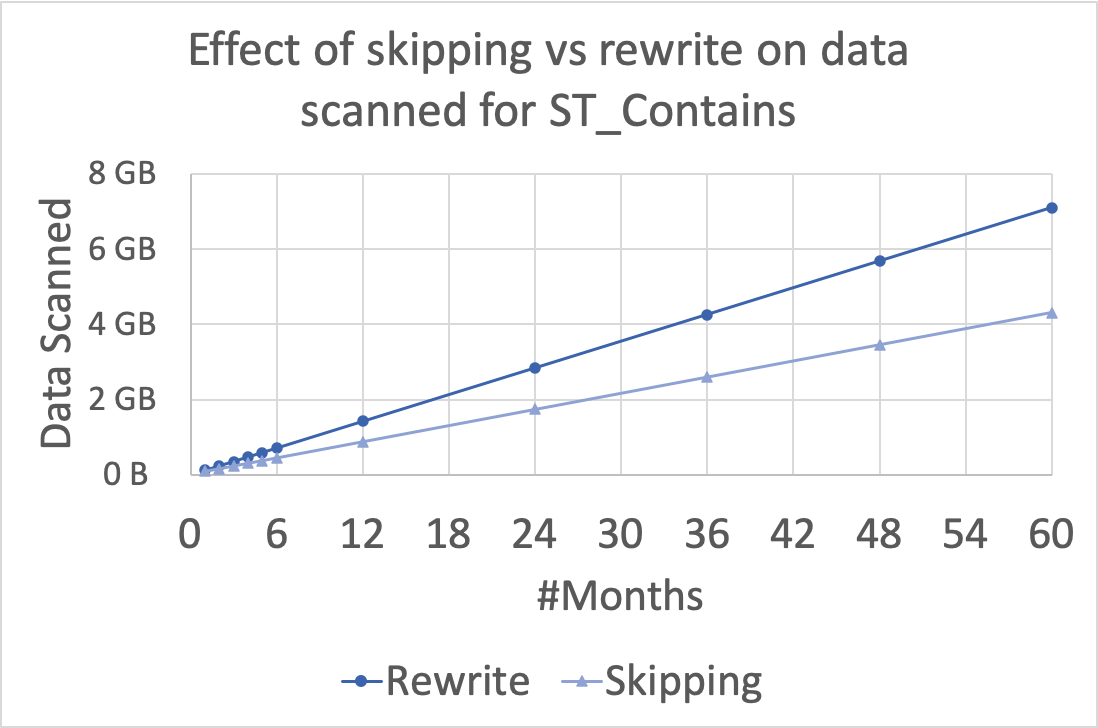}}
     \subfloat{\includegraphics[width=0.24\textwidth]{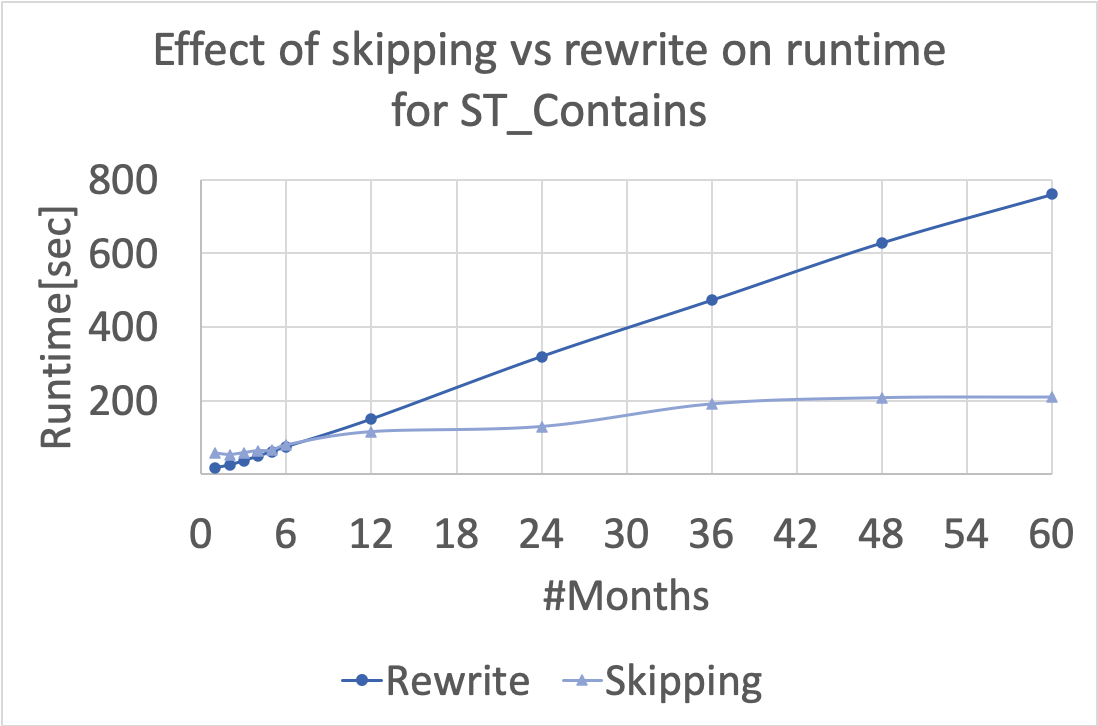}}
    \caption{Effects of skipping versus rewrite approach for ST\_Contains applied to the weather dataset with varying time window sizes}
    \label{fig:geospatialskippingvsrewrite}
\end{figure}

\subsection{Prefix/Suffix Matching}
SQL supports pattern matching using the LIKE operator, supporting single and multi-character wildcards. 
We added prefix and suffix indexes to support predicates of the form LIKE 'pattern\%' and LIKE '\%pattern' respectively. 
The indexes accept a length as a parameter and store a list of distinct prefixes (suffixes) appearing in each object. 
This is more efficient and results in smaller indexes compared to value list when a column's prefixes/suffixes are repetitive. 
\footnote{A trie based implementation is a topic for further work.}

In figure \ref{fig:patternmatching} we present the skipping effectiveness indicators for prefix/suffix matching on the db\_name column and prefix matching on the http\_request column of the cloud database logs dataset. 
For the db\_name column we stored prefixes and suffixes of length 15, and for the http\_request column we stored prefixes of length 20. Note the average column lengths for these columns are much higher. 
We generated a workload for each index consisting of 50 queries. For the prefix workloads, each query has a LIKE 'pattern\%' predicate, where the pattern is a random column value in the dataset with prefix of random size up to the column value length. The suffix workload is generated similarly.

Overall the aim is to bring the scanning factor as close as possible to the selectivity. The extent to which this is possible depends on how close we can bring the layout and metadata factors to 1 (equation \ref{eqn:factorsequation}). Despite relatively low layout factors (layout was not done according to the queried columns), good skipping is achievable. All indexes shown here achieve metadata factor close to 1, despite storing only prefixes/suffixes, and give a range of beneficial scanning factors.

\begin{figure}[h]
\centering
\includegraphics[width=3in]{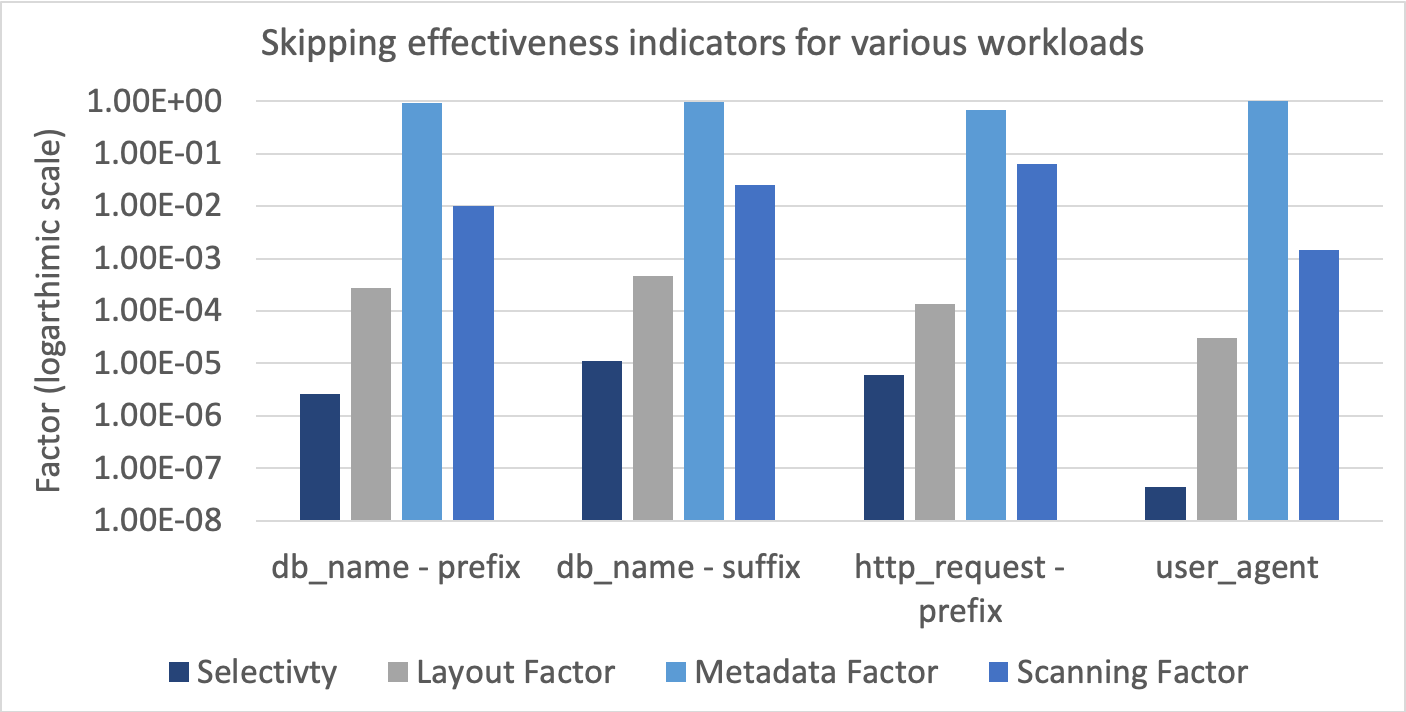}
\caption{Skipping effectiveness indicators for prefix/suffix and format specific user agent indexes. For selectivity and scanning factors lower is better, for metadata and layout factors higher is better. With approximately equal metadata factors, a highly selective (selectivity closer to 0) user\_agent workload makes up for a significantly lower layout factor, achieving the best scanning factor overall. All indexes are beneficial, achieving between 1/1000 and 1/10 scanning factors.}
\label{fig:patternmatching}
\end{figure}

In figure \ref{fig:metadatafactorprefixlength} we show the effects of increasing the prefix length in terms of skipping indicators as well as metadata size. In this case we generated a different random workload for the db\_name column with 20 queries\footnote{The selectivity is slightly different from that shown in figure \ref{fig:patternmatching} because the workload is a different set of queries.}.    
Here the selectivity and layout factors are fixed so the scanning factor is inversely proportional to the metadata factor. According to equation 
\ref{eqn:indicators} the lowest possible scanning factor is around $10^{-2}$. 
We achieve this for prefix length 15 with an order of magnitude smaller metadata compared to a value list index.

\begin{figure}[h]
\centering
\includegraphics[width=3in]{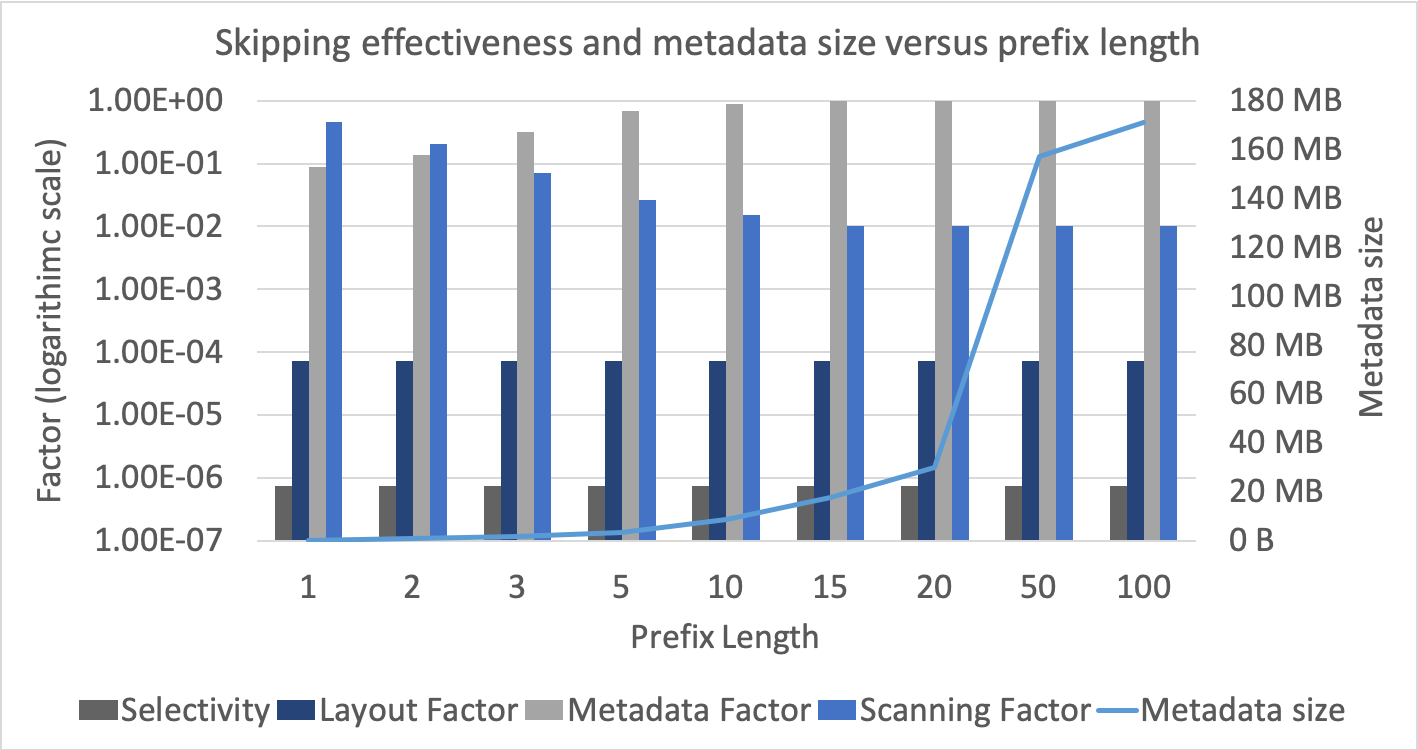}
\caption{Skipping effectiveness indicators and metadata size for a prefix index on the db\_name column w.r.t. prefix length. 
}
\label{fig:metadatafactorprefixlength}
\end{figure}

\subsection{Format Specific Indexing} \label{formatskipping}
As is typical for log analytics, many columns in our logs datasets e.g. db\_name, http\_request have additional application specific (nested) structure not captured by prefix/suffix indexes, such as hierarchical paths and parameter lists.
We show how to index such columns, avoiding the need to add new data columns, which is often not feasible for large and fast growing data.

We indexed the user\_agent column\cite{rfc7231} of both datasets to track the history of malicious http requests. Our extensible framework enables easy integration with open source tools.
We used the Yauaa library\cite{yauaa}, benefitting from its accurate client identification\cite{uablog},  its handling of idiosyncrasies in the format, and its keeping up to date with frequent client changes. 
The library parses a user agent string into a set of field name-value pairs. To generate the metadata, we parsed out the agent name field, and stored a list of names per object.
We also implemented the getAgentName UDF. The query below retrieves all malicious http requests in the log.
\begin{verbatim}
SELECT * FROM storagelogs
WHERE getAgentName(user_agent)=`Hacker'
\end{verbatim}
In figure \ref{fig:patternmatching} we show the skipping effectiveness indicators for this index, using a workload consisting of 50 queries, where for each query we chose a random agent name appearing in the dataset. 
This highly selective workload enables very good skipping even with low layout factor. 

\section{Related Work}
\label{sec:related}
{\bf Hive style partitioning} 
partitions data according to certain attributes, encoded as metadata in filenames.
Spark/Hadoop use this metadata for partition pruning.
Using this technique alone is inflexible since only one hierarchy is possible, changing the partitioning scheme requires rewriting the entire dataset when using object storage (which has no rename operation), and range partitioning is not supported. 
Our framework for extensible data skipping is complementary to this technique.

{\bf Parquet and ORC} support min/max metadata stored in file footers, as well as bloom filters\cite{parquet,orc}. Both support dictionary encodings which provide some of the benefit of our value list indexes. Note that these encodings are primarily designed to achieve compression, so in some cases other encodings are used instead, compromising skipping\cite{braams2018predicate}. 
Both formats require all objects to be partially read to process a query, 
and footer processing is not read optimised. 
Neither format allows adding metadata to an existing file, whereas our approach allows dynamic indexing choices.
Parquet allows user defined predicates as part of a Filter API, however this is designed to work with existing metadata only. Since query engines have not exposed similar APIs this does not achieve extensible skipping.  

{\bf Data skipping} 
Min/max metadata, also known as synopsis and zone maps, is commonly used in commercial DBMSs \cite{raman2013db2, ziauddin2017dimensions} and some data lakes\cite{databricks}. 
Other index types have been explored in research papers e.g. storing small materialized aggregates (SMAs) per object column such as min, max, count, sum and histograms\cite{sma}. 
Brighthouse\cite{brighthouse} defines a data skipping index similar to Gap List. Their Character Map index could be easily defined using our APIs.
Recently range sets (similar to our gap lists) have been proposed to apply data skipping to queries with joins\cite{dips}.

{\bf Data layout research} Many efforts optimize data layout to achieve optimal skipping e.g.\cite{sun2014fine, aqwa, shanbhag2017robust, quilts}. We survey those most relevant.
The fine grained approach\cite{sun2014fine} adopts bit vectors as the only supported metadata type, where 1 bit is stored per workload feature. To obtain a list of features one needs to analyze the workload, inferring subsumption relationships between predicates and applying frequent itemset mining. This approach does not work well when the workload changes. To handle a UDF, the user needs to implement a subsumption algorithm for it, although this aspect is not explained in the paper. 
On the other hand, our framework enables defining a feature based (bit vector) metadata index, allowing feature based data skipping when applicable.  

Both AQWA\cite{aqwa} and the robust approach\cite{shanbhag2017robust} address changing workloads by building an adaptive kd-tree based partitioner which exploits existing workload knowledge and is updated when as the workload changes. AQWA focuses on geospatial workloads only whereas the robust approach handles the more general case. Data layout changes are made when beneficial according to a cost benefit analysis. kd-trees apply to ordered column types, and generate min/max metadata only. Other layout techniques are needed to handle categorical data and application specific data types such as server logs and images.

{\bf Extensible Indexing} 
Hyperspace defines itself as an extensible indexing framework for Apache Spark\cite{hyperspace}, although at the time of this writing it only supports covering indexes which require duplicating the entire dataset, and does not include any data skipping (chunk elimination) indexes.
The Generalized Search Tree (GiST) \cite{HNP95,kornacker} 
focused on generalizing inverted index access methods 
with APIs such that new access methods can be easily integrated into the core DBMS supporting efficient query processing, concurrency control and recovery. Our work focuses on data skipping for big data where classical inverted indexes are not appropriate, and a different set of extensible APIs is needed. 
 
{\bf Applications}
Prior work addressed specific applications such as geospatial analytics\cite{geosurvey,aqwa}, and range and k nearest neighbour (kNN) queries for metric functions e.g.\cite{Ciaccia,chavez2000effective} without providing general frameworks.

\section{Conclusions}
\label{sec:conclusions}

Our work is the first extensible data skipping framework, allowing developers to define new metadata types and supporting data skipping for queries with arbitrary UDFs. Moreover our work enjoys the performance advantages of consolidated metadata, is data format agnostic, and has been integrated with Spark in several IBM products/services. 
We demonstrated that our framework can provide significant performance and cost gains while adding relatively modest overheads, and can be applied to a diverse class of applications, including geospatial and server log analytics.
Our work is not inherently tied to Spark and could be integrated in any system with the ability to intercept the list of objects to be retrieved. 
Further work includes integration into additional SQL engines and automatic index selection.

\section{Acknowledgements}
\label{sec:ack}
The authors would like to thank Ofer Biran, Michael Factor and Yosef Moatti for their close involvement in this work and for providing valuable review feedback. Thanks to Linsong Chu, Pranita Dewan, Raghu Ganti and Mudhakar Srivatsa for collaboration on the geospatial integration, and to Michael Haide, Daniel Pittner and Torsten Steinbach for fruitful long term collaboration. Thanks to Guy Gerson for involvement in the initial stages of this work.  

This research was partially funded by the EU Horizon 2020 research and innovation programme under grant agreement no. 779747.

\bibliographystyle{abbrv}
\bibliography{extensible_ds}  

\newpage

\appendices
\section{Formal Description and Proofs}
\label{appendix:correctness}

We point out that negation of an expression $e$ can be handled if we can construct a Clause representing $\neg e$. \begin{definition}
Let $c$ be a Clause that represents an expression $e$, we say that a Clause $c_e^*$ is a {\bf negation of $c$ with respect to $e$} if $c_e^* \wr \neg e$
\end{definition}
In the worst case, our algorithm will return None, meaning that no skipping can be done. 
\setlength{\textfloatsep}{0pt}
\begin{algorithm}
	\SetKwInOut{Input}{input} 
	\SetKwInOut{Output}{output}
	\Input{an expression tree $e$ with root $v$}
	\Output{A Clause $C$ (possibly None)}
	\BlankLine
	
	\uIf{$e = AND(a,b)$}{  \tcc{Case $1$}
		Let $\phi \coloneqq \bigwedge_{\gamma \in CS(v)}\gamma$ \\
		Run the algorithm recursively on $a$ and $b$ and denote the result by $\alpha,\beta$ respectively\\
		Return $\alpha \wedge \beta \wedge \phi$ \
	}\uElseIf{$e = OR(a,b)$}{  \tcc{Case $2$}
		Let $\phi \coloneqq \bigwedge_{\gamma \in CS(v)}\gamma$ \\
		Run the algorithm recursively on $a$ and $b$ and denote the result by $\alpha,\beta$ respectively\\
		Return $(\alpha \vee \beta)  \wedge \phi$ \
	}\uElseIf{$e = NOT(a)$}{  \tcc{Case $3$}
		Run the algorithm recursively on $a$, denote the result by $\alpha$ \\
		\uIf{$\alpha$ can be negated with respect to $a$}{
			Return $\alpha_a^*$ \
		}\Else{
			Return None \
		}
	}\uElse{ \tcc{Case $4$} 
		Return  $\bigwedge_{\gamma \in CS(v)}\gamma$ \
	}
\caption{Merge-Clause}
\label{alg:mergeclause}
\end{algorithm}	
\begin{algorithm}
	\SetKwInOut{Input}{input} 
	\SetKwInOut{Output}{output}
	\Input{a boolean expression $e$, a sequence of filters $f_1,...,f_n$}
	\Output{A Clause (possibly None) $c$}
	\BlankLine
	Apply $f_1,...,f_n$ to $e$\\
	Run $Merge-Clause(e)$ and return the result\
\caption{Generate-Clause}
\label{alg:generateclause}
\end{algorithm}	
\subsection{Correctness}
Given a query $Q$ with ET $e$, we apply algorithm \ref{alg:generateclause} to achieve a Clause $C$ using the filters defined using our extensible APIs and registered in our system. We show that $C \wr e$. Therefore we can safely skip all objects whose metadata does not satisfy $C$. 
\begin{remark}
A good perspective of how extensibility is achieved is by viewing each extensible part's role: \textbf{metadata types} stand for \textbf{what is the collected metadata}, \textbf{filters} stand for \textbf{how to utilize the available metadata on a given query}, and \textbf{metadata stores} stand for \textbf{how the metadata is stored}.
\end{remark}

\begin{thm}
\label{correctness-thm}
Let $e$ denote a boolean expression, and $f_1,...,f_k$ denote a sequence of $filters$.
Denote by $C$ the output of algorithm \ref{alg:generateclause}  on $e$ with $f_1,...,f_k$. Then $C \wr e$.
\end{thm}

\subsection{Proof of theorem \ref{correctness-thm}}
To prove the theorem, we will use the following lemmas:

\begin{lemma}
\label{conj-lemma}
Let $e$ denote a boolean $expression$, let $c_1,c_2$ s.t. $c_1 \wr e \wedge c_2 \wr e$.
Then $(c_1 \wedge c_2) \wr e$.
\end{lemma}
\begin{proof}
Assume the stated assumptions. we will show that  $(c_1 \wedge c_2) \wr  e$ by definition: let $o \in U$ s.t. $\exists r\in o.e(r)=1$. then - since $c_1 \wr e$ we get $c_1(o)=1$, identically we get $c_2(o)=1$, thus $c_1(o)=1 \wedge c_2(o)=1 \implies (c_1\wedge c_2)(o)=1$.
\end{proof}

\begin{lemma}
\label{conj-multi-lemma}
Let $e_1,e_2$ denote a pair of boolean $expressions$, let $c_1,c_2$ s.t. $c_1 \wr e_1 \wedge c_2 \wr e_2$.
Then $(c_1 \wedge c_2) \wr (e_1\wedge e_2)$.
\end{lemma}
\begin{proof}
Assume the stated assumptions and let $o \in U$ s.t. $\exists r \in o.(e_1 \wedge e_2)(r)$, we will show that $(c_1 \wedge c_2)(o)$: in particular, $e_1(r)$, which implies $c_1(o)$. identically we get $c_2(o)$, thus $c_1(o) \wedge c_2(o) \implies (c_1 \wedge c_2)(o)$.
\end{proof}

\begin{lemma}
\label{disj-multi-lemma}
Let $e_1,e_2$ denote a pair of boolean $expressions$, let $c_1,c_2$ s.t. $c_1 \wr e_1 \wedge c_2 \wr e_2$.
Then $(c_1 \vee c_2) \wr (e_1 \vee e_2)$.
\end{lemma}
\begin{proof}
Assume the stated assumptions and let $o \in U$ s.t. $\exists r \in o.(e_1 \vee e_2)(r)$, we will show that $(c_1 \vee c_2)(o)$: in particular, if $e_1(r)$ then $c_1(o)$, else we get $e_2(r)$, which implies $c_2(o)$, thus we get $c_1(o)  \vee c_2(o) \implies (c_1 \vee c_2)(o)$ 
\end{proof}

\begin{remark}
The above-mentioned lemmas can easily be re-stated and re-proved for an arbitrary number of expressions, by a simple induction. we omit these parts and from now we will use the lemmas as if stated for an arbitrary number of expressions. 
\end{remark}

\begin{lemma}
\label{correctness-lemma}
Let $e$ denote a boolean $expression$, denote by $T_e$ the expression tree rooted at $e$.
Assume the following holds:
\begin{assumption} $\forall v \in T_e \forall c\in CS(v):c \wr v$.
\label{repr-assumption}
\end{assumption}
Denote by $C$ the output of Algorithm~\ref{alg:mergeclause} on $e$, then $C \wr e$.
\end{lemma}

\newlist{casesp}{enumerate}{3}
\setlist[casesp,1]{label=\textbf{\emph{Case~\arabic*:}},ref=\arabic*}
\setlist[casesp,2]{label=\textbf{\emph{Case~\thecasespi.\alph*:}}}
\setlist[casesp,3]{label=\textbf{\emph{Case~\thecasespi.\thecasespii.\roman*:}}}

\begin{proof}
By full induction on $d$ - the depth \footnote{in this case the depth is defined as the maximum length (in edges) of a path from the root ($T_e$) to a $\{ \vee , \wedge , \neg \}$  node, comprised of $\{ \vee , \wedge , \neg \}$ nodes only, so for example the depth of $(a+b < 2) \wedge (c < 5)$ is $1$}.we will assume WLOG that all $\{ \vee , \wedge , \neg \}$ nodes are of degree $ \le 2$. 
\begin{casesp}
\item \textbf{(Base case, $d=0$)}
In this case, $e$ is a single boolean operator, so case $4$ of Algorithm~\ref{alg:mergeclause} is applied. By our assumption, $\forall c \in CS(e).c \wr e$, by lemma~ \ref{conj-lemma} , we get $\bigwedge_{\gamma \in CS(e)}\gamma \wr e$, and indeed this is the output in this case.

\item \textbf{(Induction Step)}
 Let $d \in \mathbb N^{+}$ and assume the claim holds for all $k \in \{0 ... d-1\}$. since $ d > 0$, cases ${1,2,3}$ of Algorithm~\ref{alg:mergeclause} are the only options.
\begin{casesp}
\item \textbf{\emph{if $e = AND(a,b):$}} 
in this case, Algorithm~\ref{alg:mergeclause} is called again on $a,b$, use $\alpha, \beta$ from Algorithm~\ref{alg:mergeclause}'s notation. $a,b$ are both expressions of depth strictly smaller than $d$, so by the inductive hypothesis we have $\alpha \wr a$ and $\beta \wr b$ ; by lemma~ \ref{conj-multi-lemma} we get $(\alpha \wedge \beta) \wr (a \wedge b)$.
by lemma~ \ref{conj-lemma}  and from Assumption~\ref{repr-assumption} we get  $(\phi = \bigwedge_{\gamma \in CS(e)}\gamma )\wr e$. applying lemma~ \ref{conj-lemma} again we get 
$(\alpha \wedge \beta \wedge \phi) \wr e$, and indeed this is the output in this case.
\item \textbf{\emph{if $e=OR(a,b):$}}
in this case, Algorithm~\ref{alg:mergeclause} is called again on $a,b$, use $\alpha, \beta$ from Algorithm~\ref{alg:mergeclause}'s notation. $a,b$ are both expressions of depth strictly smaller than $d$, so by the inductive hypothesis we have $\alpha \wr a$ and $\beta \wr b$; 
by lemma~\ref{disj-multi-lemma} we get $(\alpha \vee \beta) \wr e$. by lemma~ \ref{conj-lemma}  and from Assumption~\ref{repr-assumption} we get  $(\phi = \bigwedge_{\gamma \in CS(e)}\gamma )\wr e$. applying lemma~ \ref{conj-lemma} again we get 
$((\alpha \vee \beta) \wedge \phi) \wr e$, and indeed this is the output in this case.
\item \textbf{\emph{if $e=NOT(a)$:}}
in this case, Algorithm~\ref{alg:mergeclause} is called again on $a$, and the result is denoted as $\alpha$.
\item \textbf{\emph{if $\alpha$ can be negated with respect to $a$:}} Algorithm~\ref{alg:mergeclause} returns $\alpha_a^*$, and by definition $\alpha_a^* \wr \neg a = e$
\item \textbf{\emph{if $\alpha$ CAN NOT be negated with respect to $a$:}} $None$ is returned, which represents \emph{any} expression.
\end{casesp}
\end{casesp}

\end{proof}

We are now ready to prove Theorem~\ref{correctness-thm}:

\begin{proof}[Proof of Theorem~\ref{correctness-thm}]
From Algorithm~\ref{alg:mergeclause}'s assumptions we know that $f_1 ,...,f_n$ are \emph{filters}, thus Assumption~\ref{repr-assumption} holds.
Thus, correctness follows from Lemma~\ref{correctness-lemma}.
\end{proof}

\section{Indexing Statistics}
\label{sec:indexstats}
\begin{table}[!ht]
\centering
\begin{threeparttable}
\caption{Indexing Statistics}
\begin{tabular}{|c|c|c|c|c|} \hline
{\bf Index}&{\bf Col} &{\bf Num.} & {\bf MD} & {\bf Indexing}\\ 
{\bf Type}&{\bf Size} &{\bf Objects} & {\bf Size} & {\bf Time}\\ 
& (GB) & & (MB) & (min) \\ \hline
ValueList & 6.73  & 4000 & 163.2 & 9.4  \\ \hline
BloomFilter & 6.73  & 4000 & 67.8 & 10.0  \\ \hline
Hybrid & 6.73 & 4000 & 40.4 & 9.7 \\ \hline
Prefix(15) & 6.73 & 4000 & 17.0 & 10.0 \\ \hline
Suffix(15) & 6.73 & 4000 & 35.5 & 10.0 \\ \hline
Value List & 0.39  & 4000 & 34.2 & 8.5 \\ \hline
BloomFilter & 0.39  & 4000 & 38.9 & 9.4 \\ \hline
Hybrid & 0.39  & 4000 & 34.3 & 10.0 \\ \hline
Formatted\tnote{1} & 0.72 & 4000 & 0.27 & 15 \\ \hline
MinMax & 0.56 & 4000 & 0.125 & 0.97 \\ \hline
MinMax & 12.16 & 8192 & 0.38 & 1.2 \\ \hline
\end{tabular}
\label{table:indexstats}
\begin{tablenotes}
\item [1]Identifies malicious requests using user\_agent column - see section \ref{formatskipping}.
\end{tablenotes}
\end{threeparttable}
\end{table}

\section{Example data skipping index} \label{indexexample}
The following is a simplified version of the user\_agent index from section \ref{formatskipping}.
We registered a UDF in Spark which uses the Yauaa library\cite{yauaa} to extract the user agent name from a user\_agent string.
\begin{lstlisting}
import nl.basjes.parse.useragent._
object UserAgentUDF {
  // Define the analyzer at the object level 
  val analyzer = UserAgentAnalyzer.newBuilder().build()
  val getAgentName = (userAgent: String) => {
    val res = analyzer.parse(userAgent)
    .get(UserAgent.AGENT_NAME)
    res.getValue
  }}
\end{lstlisting}
The UDF registration is done in the main code using
\begin{lstlisting}
spark.udf.register("getAgentName", UserAgentUDF.getAgentName)
\end{lstlisting}
The user\_agent index collects a list of distinct agent names, therefore, we reuse the Value List  {\tt MetaDataType} (representing a set of strings) and its {\tt Clause} as well as the translation for both.

\subsection{Index Creation}
\begin{lstlisting}
case class UserAgentNameListIndex(column : String) extends Index(Map.empty, column) {
    def collectMetaData(df: DataFrame): MetadataType = {
    ValueListMetaData(
    	    df.select(getAgentName(col("user_agent")))
      		.distinct().collect().map(_.getString(0)))
}}
\end{lstlisting}

\subsection{Query evaluation}
The following filter identifies the query pattern appearing in section \ref{formatskipping}.
\begin{lstlisting}
case class UserAgentNameFilter(col:String) extends BaseMetadataFilter {
 def labelNode(node:LabelledExpressionTree): Option[Clause] = {
node.expr match {
         case EqualTo(udfAgentName: ScalaUDF,v : Literal) if isUserAgentUDF(udfAgentName, col) =>
        		Some(ValueListClause(col, Array(v.value.toString)))	
          case _ => None
    }}       
\end{lstlisting}
The function {\tt isUserAgentUDF} checks for a match in the ET with the {\tt getAgentName} UDF.\\

\end{document}